\documentclass[11pt]{article}
\usepackage{graphicx}
\usepackage{amsmath}
\usepackage{amssymb}
\usepackage{xspace}
\usepackage{ifthen}

\def\ls{\log^* n}

\def\sse{\subseteq}

\def\M{{\mathcal{M}^{(2)}}}
\def\MK{{\mathcal{M}^{(2,k)}}}

\def\F2{\mathcal{F}^{(2)}}
\def\C{\mathcal{C}}
\def\A{\mathcal{A}}
\def\dist{\mathcal{F}}
\def\ivvi{{\bar{\phi}_i}}

\def\algfull{{\sc Optimal Price Scaling (OPS)}\xspace}
\def\alg{{\sc OPS}\xspace}

\def\algoptfull{{\sc Black-Box Reduction (BBR)}\xspace}

\newcommand{\prob}[2][]{\text{\bf Pr}\ifthenelse{\not\equal{}{#1}}{_{#1}}{}\!\left[#2\right]}
\newcommand{\expect}[2][]{\text{\bf E}\ifthenelse{\not\equal{}{#1}}{_{#1}}{}\!\left[#2\right]}

\newcommand{\bid}{b}
\newcommand{\bids}{{\mathbf \bid}}
\newcommand{\bidsmi}{{\mathbf \bid}_{-i}}
\newcommand{\bidi}[1][i]{\bid_{#1}}

\newcommand{\val}{v}
\newcommand{\vals}{{\mathbf \val}}
\newcommand{\valsmi}{{\mathbf \val}_{-i}}
\newcommand{\vali}[1][i]{{\val_{#1}}}

\def\v2{\val^{(2)}}

\newcommand{\price}{p}
\newcommand{\prices}{p}
\newcommand{\pricei}[1][i]{{\price_{#1}}}

\newtheorem{theorem}{Theorem}[section]
\newtheorem{lemma}		[theorem]	{Lemma}

\newtheorem{definition}	[theorem]	{Definition} 

\newenvironment{proof}{\noindent {\em {Proof:}}}{$\blacksquare$\vskip \belowdisplayskip}

\title{Near-Optimal Multi-Unit Auctions with Ordered Bidders}

\author{
 Elias Koutsoupias\thanks{Partially supported by the
    European Union Seventh Framework Programme FP7, grant 284731
    (UaESMC), and by the ESF-NSRF research program Thales (AGT).} \\
        University of Oxford and \\
	University of Athens\\
        \tt{elias@cs.ox.ac.uk}
\and Stefano Leonardi\thanks{This work was partially supported from EU ERC research grant PAAI (Practical Approximation Algorithms)}\\
	Sapienza University of Rome\\
        \tt{leon@dis.uniroma1.it}
\and Tim Roughgarden\thanks{Supported in part by NSF grant
CCF-1016885, an ONR PECASE Award, and an AFOSR MURI grant.}\\
		    Stanford University\\
        \tt{tim@cs.stanford.edu}
}

\begin{document}

\maketitle

\begin{abstract}
We construct prior-free auctions with constant-factor approximation
guarantees with ordered bidders, in both unlimited
and limited supply settings.  We compare the expected revenue of our
auctions on a bid vector to the monotone price benchmark, the 
maximum revenue that can be
obtained from a bid vector using supply-respecting prices that are
nonincreasing in the bidder ordering and bounded above by the
second-highest bid.
As a consequence, our auctions are simultaneously near-optimal 
in a wide range of Bayesian multi-unit environments.
\end{abstract}

\section{Introduction}
\label{sec:intro}

The goal in prior-free auction design is to design auctions that have
robust, input-by-input performance guarantees.  Traditionally,
auctions are evaluated using average-case or Bayesian analysis,
and expected auction performance is optimized with respect to a prior
distribution over inputs (i.e., bid vectors).  The Bayesian versions of
the problems we consider are completely solved~\cite{M81}.  Worst-case
guarantees are desirable when, for example, good prior information is
expensive or impossible to acquire, and when a single auction is to be
re-used several times, in settings with different or not-yet-known
input distributions. 

Prior-free auctions were first studied by
Goldberg et al.~\cite{G+06,GHw99}.  They focused on symmetric settings,
where goods and bidders are identical, and sought auctions with
expected revenue close to the fixed-price benchmark~$\F2$,
defined as the maximum revenue that can be obtained from a given bid
vector by offering every bidder a common posted price (i.e.,
take-it-or-leave-it offer) that is at most
the second-highest bid.
Goldberg et al.~\cite{G+06} showed that no auction has expected
  revenue more than a $\approx .42$ fraction of $\F2$ for every bid
  vector, and constructed auctions with expected revenue
at least a constant fraction of this benchmark on every input.
See Hartline and Karlin~\cite{HK07} for a survey of further work in
this vein.

Hartline and Roughgarden~\cite{HR08} proposed a framework for defining
meaningful performance benchmarks much more generally --- when bidders
or feasibility constraints are asymmetric, and for objective functions
other than revenue.
The first step of this framework is a ``Bayesian thought experiment''
--- if bidders' valuations were drawn from a prior distribution in some
class, what would the optimal auction be?
The second step is to characterize the collection~$\C$ of all optimal
auctions that can arise, ranging over all permissible prior distributions.
Finally, given a bid vector~$\bids$, the performance benchmark is
defined as the maximum objective function value obtained by an auction
in~$\C$ on the input~$\bids$. 
This framework regenerates the $\F2$ benchmark (modulo the technically
necessary upper bound on prices) and has been used for several
other objective functions and asymmetric
environments~\cite{DH09,HR08,HR09,HY11,LR12}.
Every benchmark generated by this framework is
automatically well motivated in the following sense: if the
performance of an auction is within a constant factor of such a
benchmark for every input, 
then in particular it is simultaneously near-optimal in
every Bayesian environment with valuations drawn from one of the
permissible prior distributions.\footnote{This weaker goal of good
{\em prior-independent} auctions can also be studied in its own
right~\cite{D+11,DRY10,RTY12}.  See~\cite{ADMW13,CM11,LRS09} for other
interpolations between average-case and worst-case analysis of auctions.}

Leonardi and Roughgarden~\cite{LR12} studied the design and analysis
of prior-free digital goods (i.e., unlimited supply) auctions
with asymmetric bidders.  They pointed out that the framework
in~\cite{HR08} can be applied successfully to non-identical bidders
only if sufficient qualitative information about bidder asymmetry
is publicly known.  They proposed a model of {\em ordered bidders}.
Earlier bidders are in some sense expected to have higher valuations.
This information could be derived from, for example, zip codes,
eBay bidding histories, credit history, previous transactions
with the seller, and so on.  
Leonardi and Roughgarden~\cite{LR12} defined the 
{\em monotone price benchmark}~$\M(\bids)$ for every bid vector
$\bids$ as the maximum revenue obtainable via a monotone
price vector --- meaning prices are nonincreasing in the bidder
ordering --- in which every price is at most the second-highest
bid.\footnote{This benchmark was also considered earlier, with a
  different motivation and application, by Aggarwal and
  Hartline~\cite{AH06}.}
The value of this benchmark is always at least that of the fixed-price
benchmark~$\F2$, and can be a factor of $\Theta(\log n)$ larger,
where~$n$ is the number of bidders.
Essentially by construction, a digital goods auction that always has
revenue at least a constant fraction of~$\M$ is simultaneously
near-optimal in every Bayesian environment with ordered distributions
(where monopoly prices are nonincreasing in the bidder ordering), or
when the valuation distribution of each bidder stochastically dominates
that of the next one in the ordering (see~\cite{LR12} for details).
Examples include uniform distributions with intervals~$[0,h_i]$ and
nonincreasing $h_i$'s; exponential 
distributions with nondecreasing rates; Gaussian distributions with
nonincreasing means; and so on.  
The main result in~\cite{LR12} is a prior-free digital goods auction
with ordered bidders with expected revenue~$\Omega(\M(\bids)/\ls)$
for every input~$\bids$, where~$n$ is the number of
bidders and~$\ls$ denotes the number of times that the $\log_2$ operator can
be applied to~$n$ before the result drops below a fixed constant.\footnote{Aggarwal and Hartline~\cite{AH06} previously
  obtained an incomparable guarantee of
$\Omega(\M(\bids)) - O(h \log \log \log h)$, where 
$h$ is the ratio between the maximum and minimum bids.}

\subsection{Our Results} 

We give the first digital goods auction that is~$O(1)$-competitive
with the monotone price benchmark~$\M$.  Our auction is simple and
natural.  It follows the standard approach of randomly partitioning the
bidders into two groups, using one group of bidders to set prices for
the other.  We restrict prices to be (essentially) all powers of a certain
constant, but otherwise our prices are simply the optimal monotone
ones for the first bidder group.  Finally, to handle inputs where the
monotone price benchmark derives most of its revenue from a small
number of bidders, with constant probability we invoke an auction that
is~$O(1)$-competitive with the fixed-price benchmark~$\F2$.

We extend our results to multi-unit auctions, where the number of
items~$k$ can be less than the number of bidders.  We consider the
analog~$\MK$ of the monotone price benchmark, which maximizes only
over (monotone) price vectors that sell at most~$k$ units.  We prove that
every auction that is~$O(1)$-competitive with the benchmark~$\MK$
implies simultaneously near-optimal for a range of Bayesian
multi-unit
environments --- roughly, those in which the (ironed) virtual
valuation functions of the bidders form a pointwise total ordering.
We also give a general reduction, showing how to build a
limited-supply auction that is~$O(1)$-competitive w.r.t.~$\MK$ from an
unlimited-supply auction that is~$O(1)$-competitive w.r.t.~$\M$.

\section{Preliminaries}\label{sec:prelim}

In a {\em multi-unit auction}, there is one seller, $n$ bidders,
and~$k$ identical items.
Each bidder wants
only one good, and has a private --- i.e., unknown to the seller ---
{\em valuation} $\vali$.
We call the special case where~$k=n$ {\em unlimited supply} or {\em
  digital goods}.
We study direct-revelation auctions, in which the bidders
report bids~$\bids$ to the seller, and the seller then decides who
wins a good and at what price.\footnote{For the questions we ask, the 
``Revelation Principle'' (see, e.g., Nisan~\cite{N07}) ensures that
  there is no loss of generality by considering only direct-revelation
  auctions.}  For a fixed (randomized)
auction, we use $X_i(\bids)$ and $P_i(\bids)$ to denote the winning
probability and expected payment of bidder~$i$ when the bid
profile is~$\bids$.
As in previous works on prior-free auction design, we consider only
auctions that are individually rational --- meaning $P_i(\bids) \le
v_i \cdot X_i(\bids)$ for every~$i$ and~$\bids$ --- and truthful,
meaning that for each bidder~$i$ and fixed bids~$\bidsmi$ by the other bidders,
bidder~$i$ maximizes its quasi-linear utility $v_i \cdot
X_i(\bidi,\bidsmi) - P_i(\bidi,\bidsmi)$ by setting $\bidi = \vali$.
Since we consider only truthful auctions, from now on we use
bids~$\bids$ and valuations~$\vals$ interchangeably.

Truthful and individually rational digital goods auctions have a nice
canonical form: for every bidder~$i$ there is a
(possibly randomized) function~$t_i(\valsmi)$ that, given the
valuations~$\valsmi$ of the other bidders, gives bidder~$i$ a
``take-it-or-leave-it offer'' at the price~$t_i(\valsmi)$.  This means
that bidder~$i$ is given a good if and only if~$v_i \ge t_i(\valsmi)$,
in which case it is charged the price $t_i(\valsmi)$.  It is clear
that every choice~$(t_1,\ldots,t_n)$ of such functions defines a
truthful, individually rational digital goods auction; conversely,
every such auction is equivalent to a choice
of~$(t_1,\ldots,t_n)$~\cite{G+06}.  A special case of
such an auction is a {\em price vector~$\prices$}, in which each~$t_i$
is the constant function~$t_i(\valsmi) = \pricei$.  When the supply is
limited (i.e., there are $k<n$ copies of the good), truthful
auctions induce functions~$(t_1,\ldots,t_n)$ with the property
that, on every input, at most~$k$ bidders win.

The {\em revenue} of an auction on the valuation profile~$\vals$ is
the sum of the payments collected from the winners.  Let $\v2$ denote
the second-highest valuation of a profile~$\vals$.
The {\em fixed-price   benchmark}~$\F2$ is defined, for each valuation
profile~$\vals$, as the maximum revenue that can be obtained from a
constant price vector whose price is at most~$\v2$:
$$
\F2(\vals) = \max_{p \le \v2} \left( \sum_{i \,:\, v_i \ge p} p \right).
$$
Now suppose there is a known ordering on the bidders, say $1,2,\ldots,n$.
The {\em monotone-price   benchmark}~$\M$ is defined analogously
to~$\F2$, except that non-constant monotone price vectors are also
permitted:
\begin{equation}\label{eq:m2}
\M(\vals) = \max_{\v2 \ge p_1 \ge p_2 \ge \cdots \ge p_n} \left(
\sum_{i \,:\, v_i \ge p_i} p_i \right).
\end{equation}
Clearly, $\M(\vals) \ge \F2(\vals)$ for every input~$\vals$.

The monotonicity and upper-bound constraints are
  enforced only in the computation of the benchmark~$\M$.  Auctions,
  while obviously not privy to the private valuations, can
  employ whatever prices they see fit.  This is natural for prior-free
  auctions and   also necessary for non-trivial results~\cite{GH03}.

Finally, when we say that an auction is {\em $\alpha$-competitive} with or has
  {\em approximation factor~$\alpha$} for a benchmark, we mean that
  the auction's expected revenue is at least a $1/\alpha$ fraction of the
  benchmark for every input~$\vals$.

\newcommand{\primal}{{\mathcal I}}
\newcommand{\start}{\text{start}}
\newcommand{\e}{\text{end}}
\newcommand{\mass}{\mu}
\newcommand{\smashed}{\hat I}
\newcommand{\I}{{\mathcal I}}
\newcommand{\hv}{{\hat \vals}}
\newcommand{\hA}{{\hat A}}
\newcommand{\hB}{{\hat B}}
\newcommand{\halgo}{\hat{\text{ALG}}}
\newcommand{\algo}{\text{ALG}}
\newcommand{\pre}[1]{\sigma(#1)}
\newcommand{\pricesR}{p}

\section{Optimal price scaling}
\label{sec:prior-free-auction}

In this section, we propose and analyze a simple auction, which we call \algfull. The basic idea is very natural: partition the sequence of bidders into two random sets, find the optimal $\M$ price vector $\pricesR$ for one part $A$ (the \emph{training part}) and offer $\pricesR$ to the other part $B$ (the \emph{test part}). We add two twists to this standard algorithmic scheme. First, we restrict the prices to be on discrete levels, so that at every level the prices remain the same, while the price drops from level to level by a constant factor $w$, where $w$ is a parameter of the algorithm (a particular value that works is $w=25$)\footnote{From our analysis of this algorithm, it follows if we, instead of fixing the prices at discrete levels, offered to the trial part the optimal price of the training part reduced by a factor $w$, we will still get a truthful algorithm with constant approximation ratio.}. The levels of the prices start at the second maximum value ${\vals^A}^{(2)}$ of the training part $A$. The second twist is that we run the algorithm with the dropping prices only with some constant probability and, with the remaining probability, we run a standard digital goods auction with constant competitive ratio. The precise description of the auction is given in Figure~\ref{fig:alg}.

\begin{figure}[h]
\hrule\medskip
\textbf{Input}: A valuation profile $\vals$ for a totally ordered
set~$N = \{1,2,\ldots,n\}$ of bidders.
\begin{enumerate}

\item With probability $1/2$, run a digital goods auction on~$\vals$ that is
  $O(1)$-competitive with respect to the benchmark $\F2$. With the
  remaining probability, run the following steps.

\item Choose a subset $A \sse N$ uniformly at random, and
  partition~$N$ into the two sets~$A$ and~$B$. Let $\vals^A$ denote the valuation profile $\vals$ in which we set the values not in $A$ to $0$, that is, 
$$\vals^A_j=
\begin{cases}
\vals_j & \text{if $j\in A$} \\
0 & \text{otherwise}.
\end{cases}
$$
Define $\vals^B$ in a similar way. All three sequences $\vals$,
$\vals^A$, and $\vals^B$ have the same length.
\item Compute an optimal monotone $\M$ price vector~$\pricesR$
  for~$A$ with prices restricted to discrete values in  $\{{\vals^A}^{(2)} /w^k \,:\, k=0,\ldots\}$. 

\item Sell items to bidders in~$B$ only, applying prices~$\pricesR$ to $\vals^B$.

\end{enumerate}
\caption{\textsf{The auction \algfull.}\label{fig:alg}}
\medskip\hrule
\end{figure}
 
We next elaborate on the steps of the auction.
In the first step, we
run an arbitrary digital goods auction that is~$O(1)$-competitive with
respect to the 
fixed-price benchmark~$\F2$.  The best-known approximation factor
is~3.12~\cite{II10}; there are also very simple auctions with
approximation factors~4~\cite{G+06} and~4.68~\cite{AMS09}.  Intuitively, this
step is meant to extract good revenue from the set of bidders with
valuations almost as high as the second-highest valuation.

The second step of the algorithm randomly partitions the bidders into a
``training set''~$A$ and a ``test set''~$B$.  Almost all prior-free
auctions have this structure, with the bidders in the training set
setting prices for those in the test set.  We want to keep the three sequences
$\vals^A$, and $\vals^B$ aligned to simplify the pricing of
the next two steps. To do this, we keep all the elements of $\vals$
in $\vals^A$ and $\vals^B$, but in $\vals^A$ we lower the elements that are not in
$A$ to 0, and similarly, we lower the elements of sequence $\vals^B$ that
are not in $B$ to 0.

Also for simplicity, we sell
(in the fourth step) only to bidders in the test set~$B$.  An obvious
optimization is to sell simultaneously to bidders in~$A$, using the bids
of~$B$; this would improve the hidden constant in our approximation
guarantee by a factor of~2.

The optimal monotone price vector is the one that maximizes the $\M$ revenue obtained from the bidders in~$A$ when prices are scaled down to the next level $\frac{{\vals^A}^{(2)}}{w^k}$. 
The scaling down of the prices won't reduce the optimal revenue on the $A$ set of bidders for more than a factor of $w$. The final step applies the prices to bidders in the test set~$B$.

The \alg auction is truthful, as each bidder faces a take-it-or-leave-it offer at a price that is independent of its reported valuation. We also note that the \alg auction can be implemented in polynomial time, as $\pricesR$ can be computed efficiently using dynamic programming.

Our main result is a prior-free approximation guarantee for the \alg
auction.

\begin{theorem}\label{thm:main}
There is a constant $c_0 > 0$ such that, for every valuation profile
$\vals$, the expected revenue of the \alg auction is at least 
$c_0 \cdot \M(\vals)$.
\end{theorem}

We outline the main ideas of the proof of the theorem here and we present the details below. Let us first make a useful assumption
\begin{definition}[First running assumption]
In the analysis of the algorithm, we will assume that the two bidders of $\vals$ with the highest valuation fall into $A$. Therefore $$\vals^{(2)}={\vals^A}^{(2)}.$$ 

This assumption holds with probability $1/4$, which will essentially increase the competitive ratio of the analysis below by a factor of 4. We remove this assumption, when we put all the pieces of the analysis together.
\end{definition}

Let us now fix some notation: 
\begin{definition} 
Define the $k$-th price level $p_{(k)}$ to be
$$p_{(k)}=\vals^{(2)}/w^k.$$ 
Let also $J_k$ denote the interval in which the prices of the auction are at the $k$-th level:
$$J_k=\{j\,:\, p_j=p_{(k)}\}.$$
\end{definition}
The intervals $J_0,\ldots, J_m$ are defined in the third step of the algorithm from $\vals^A$. If there are many optimal solutions in the third step, we fix one and we use it in the final step of the algorithm.

The intuition for the algorithm is that the  optimal revenue for the whole sequence $\vals$ is up to a constant factor equal to the expected optimal revenue from the subsequence $\vals^A$. We then need only to compare the optimal expected revenues from the subsequences $\vals^A$ and $\vals^B$. If there were many values in a pricing level $J_k$, then the random partition is expected to split almost evenly the high-valued bids between the training set $A$ and the test set $B$ and this will allow us to relate their revenue. One problem with this approach is that some levels may have few items; another more subtle problem is that even when a level has many items, we cannot easily argue that the two parts have almost the same number of high values, because there is a bias towards $A$ by the way the levels were created. 

To resolve both these issues and in order to compare the revenue of
$A$ and $B$, we use as an intermediary a set of some fixed intervals,
defined with respect to the set $\vals$  of all values (in contrast to the way
that the levels are created by the pricing of $A$). The set of fixed
intervals, which we will call \emph{primal intervals}, consists of one
collection of intervals for each pricing level. We want them to
have two crucial properties: first, \emph{every primal interval} to
have almost the same fraction (and in particular a fraction in
$[1/3,2/3]$)  in $A$ and $B$ of the bids higher than the pricing level; and second, that the revenue of $A$ is captured up to a constant factor by a set of primal intervals. 

\subsection{Primal intervals}
\label{sec:intervals}

We now define the set of primal intervals. They will be defined in terms of $\vals$, the whole sequence of values. We will need some definitions first:

\begin{definition} \ %
\begin{itemize}
\item Define $N_{\ell}$ to be the set of values greater or equal to the $\ell$-th price level.
\begin{equation} \label{eq:defN}
N_{\ell}= \{j\,:\, \vals_j\geq p_{(\ell)}\}
\end{equation}
\item Define $\mass_{\ell}(I)$, the mass of an interval $I$ at price level $\ell$, to be the number of its values that are in $N_{\ell}$ (at price level $\ell$ or higher).
$$\mass_{\ell}(I)=|I \cap N_{\ell}|.$$
\end{itemize}
\end{definition}

For each price level $\ell$, we create a set of intervals $\I_{\ell}$ which we will call \emph{primal intervals at price level $\ell$}. When we refer to the mass of a primal interval $I\in\primal_{\ell}$, we will mean $\mass_{\ell}(I)$, its mass at its price level. We will denote by $\start(I)$ and $\e(I)$ the beginning and end of an interval.

The primal intervals are defined with respect to 3 parameters: the parameter $w$ which defines the price levels in the algorithm, a parameter $\epsilon\in (0,1)$, which will be fixed later (a value of $\epsilon=1/25$ will work for the proof) and a parameter $s_0$; notice that the parameters $\epsilon$ and $s_0$ are used only in the analysis and they are not part of the algorithm.
 
We begin the construction with a fixed interval $I_o\in \I_0$ which starts at the first value and has $\mass_0(I_o)=s_0$, and we define the other intervals recursively. If $I\in \I_{\ell}$ is an interval that has already been defined, then we add the following 5 intervals every one of which has mass $(1+\epsilon)\cdot\mass_{\ell}(I)$; if there is not enough mass to create some of these 5 intervals, we still create them with the maximum possible mass to facilitate the recursive construction, we call them incomplete, and we throw them away at the end of the construction.
\begin{enumerate}
\item An interval at level $\I_{\ell}$ which starts at $\start(I)$.
\item An interval at level $\I_{\ell}$ which starts at $\e(I)$.
\item An interval at level $\I_{\ell+1}$ which starts at $\start(I)$.
\item An interval at level $\I_{\ell+1}$ which starts at $\e(I)$.
\item An interval at level $\I_{\ell}$ which \emph{ends} at $\e(I)$.
\end{enumerate}
\begin{figure}
\centering
\includegraphics{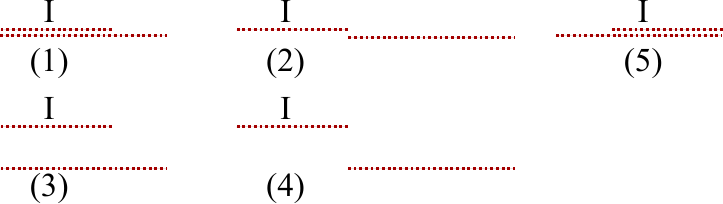}
\caption{The 5 intervals created from interval $I$. Rules 1, 2 and 5
  create intervals at the same level with $I$, while Rules 3 and 4 
  create intervals at the next level. The drawing indicates the mass of the intervals---equal to $(1+\epsilon)$ the mass of $I$---not their actual length.}
\label{fig:primalintervals}
\end{figure}
Figure~\ref{fig:primalintervals} illustrates the construction. The reason for keeping around incomplete intervals is to allow an interval at level $\ell$ to be able to create intervals at all higher levels. Specifically, we want an interval $I\in \I_{\ell}$ to create every possible interval at level $\ell+r$, $r\geq 1$, that start at the end (or beginning) of $I$ and have mass $(1+\epsilon)^r$ times the mass of $I$. For instance, even if the interval at level $\ell+r$ is incomplete, we keep it in the recursive construction to produce the potentially complete intervals at level $\ell+r+1$ by using Rules 3 and 4. We extend the idea of keeping around incomplete intervals to the original interval $I_o$: if there are not enough elements with high mass, we define $I_o$ to include the whole sequence, we call it incomplete, we use it to create recursively the other intervals, and we throw it away at the end. The mass of an incomplete interval in the construction is assumed to be the originally intended mass. For example, suppose that a complete interval $I\in \I_{\ell}$ creates an incomplete interval $I'$ which in turn creates a complete interval $I''$; the mass of $I''$ is $(1+\epsilon)^2$ times the mass of $I$, irrespectively of the actual mass of $I'$.

We now show a crucial property of the primal intervals: with probability almost 1, every primal interval contains almost equal number of elements in $A$ and $B$. More precisely, let us call a primal interval $I\in \primal_{\ell}$ \emph{balanced}, if at least a third of the elements in $I\cap N_{\ell}$ belong to $A$ and at least another third belongs to $B$. For every $\epsilon$, by selecting the parameter $s_0$ appropriately, we can make the probability of having some unbalanced primal interval arbitrarily small. Before we state and prove a useful fact:

\begin{lemma}
If $d$ is a positive constant such that $d^{\epsilon\cdot s_0}\leq 1/10$, then 
$$\sum_{\ell} \sum_{I\in\I_{\ell}} d^{\mass_{\ell}(I)} \leq 2\cdot d^{s_0}.$$
\end{lemma}
\begin{proof}
Suppose that the initial interval in the construction has mass $s\geq s_0$ and define
$$f(s)=
\sum_{\ell} \sum_{I\in\I_{\ell}} d^{\mass_{\ell}(I)},$$ to be the maximum
value attained from the intervals produced from $s$ by any sequence of values.  Essentially, we want to
bound $f(s_0)$. By the recursive construction
of the intervals, we get that 
$$f(s)\leq c^s+5f((1+\epsilon)\cdot s).$$
The lemma now follows from showing $f(s)\leq 2 d^s$ by backwards induction. For $s>n$, this is vacuously true since there are no intervals. Suppose then that $f((1+\epsilon) s)\leq 2d^{(1+\epsilon)s}$, and we bound 
\begin{align*}
f(s) &\leq d^s+5f((1+\epsilon)s) \leq d^s+10d^{(1+\epsilon)s} \\
&= d^s+10 d^s d^{\epsilon\cdot s} \leq d^s+10 d^s d^{\epsilon s_0} \leq d^s +10 \frac{1}{10} d^s=2d^s.  
\end{align*} 
\end{proof}

\begin{lemma}
\label{lemma:balance} 
There is a large enough constant $s_0$, such that for every $\vals$ and $\epsilon$, with probability at least $31/32$ every primal interval is balanced.
\end{lemma}
\begin{proof}
The probability that a primal interval $I$ is \emph{not} balanced is
at most $d^{\mass_{\ell}(I)}$ for some constant $d$ (Chernoff bound). Therefore, the
probability that some interval is not balanced is at most
$$\sum_{\ell} \sum_{I\in\I_{\ell}} d^{\mass_{\ell}(I)}\leq 2 d^{s_0},$$
which drops exponentially in $s_0$. By selecting $s_0$ large enough, we can make this probability as small as desired.  
\end{proof}

\subsection{Matching intervals}
\label{sec:matching-intervals}

To gain some intuition of the use of the primal intervals, let us
assume that we could associate with every interval $J_k$ a primal
interval $I_k\in\primal_k$ which is also a subset of $J_k$. Let us
further assume that among the possible $I_k$'s, we always select the
interval with maximum length. There are two facts that allow us to
relate the revenues from $B$ and from $A$: the first is that by the
construction of the primal intervals, the length of $I_k$ contains at
least $1/(1+2\epsilon)$ of $J_k$; the second is that the interval $I_k$ is
balanced. The latter shows that the revenue from $B$ in $I_k$ is a
constant fraction of the revenue from $N$ in $I_k$, which in turn is a
constant fraction from the revenue from $A$ in $J_k$.  The problem is
that there may be intervals $J_k$ that contain no primal interval in
$\primal_k$. To deal with this problem, we  use a charging argument which shows that the revenue lost from such intervals is not significant. 

To turn this idea into a concrete proof, we give a process that attempts to match every interval $J_k$ to some primal interval in $\primal_k$; if the process succeeds, we will call $J_k$ \emph{matched} and we will denote its matching primal interval by $I_k$; if the process does not succeed for some interval $J_k$ we will call $J_k$ \emph{unmatched}. We first match every interval $J_k$ that contains some primal interval at the same price level.
\begin{definition}
An interval $J_k$ is called nice, if there is an primal interval in $I\in\primal_k$ such that $\start(J_k)\leq \start(I) \leq \e(I) \leq \e(J_k)$. We match $J_k$ with the maximum such interval $I$ and we call it $I_k$.
\end{definition}
We process the remaining intervals from left to right. To bootstrap
the process, we make the assumption that the first interval $J_0$ is nice;  we will revisit and remove this assumption at the end of the proof.

\begin{definition}[Second running assumption] \label{def:assumption2}
Interval $J_0$ is nice, i.e., it contains a primal interval of $\primal_0$.
\end{definition}

Starting with each nice interval, we process its following intervals from left to right. We attempt to find a matching interval by creating a sequence of candidate intervals. The definition of the sequence of candidate matching intervals will be also useful later in the presentation of the main argument. Figure~\ref{fig:matching} illustrates the process which is defined in detail below: 

\begin{definition}[Matching]
For every interval $J_t$, we associate a non-empty sequence of \emph{candidate matching intervals}, all of which are at price level $t$. If $J_t$ is a nice interval, the sequence contains only the matching interval $I_t$. For every other interval $J_t$, the sequence of candidate matching intervals is created as follows:

The first candidate matching interval will be called $I'_t$ plays also a central role in the proof and it is based on the candidate matching intervals of the previous interval $J_{t-1}$. In particular, if the previous interval is matched, then $I'_t$ is the interval produced from $I_{t-1}$ by Rule 4. Otherwise, if the previous interval is unmatched, $I'_t$ is the interval produced from $I'_{t-1}$ by Rule 3.

If $\e(I'_t)>\e(J_t)$, the sequence of candidate matching intervals contains only $I'_t$ and we leave $J_t$ unmatched. 

Otherwise, the sequence of candidate matching intervals is produced from this first interval $I'_t$ by repeated applications of Rule 1, as long as the produced intervals do not extend beyond $J_t$ (i.e., they end at or before $\e(J_t)$). 

We match $J_t$ with the maximum interval in the candidate matching sequence, which we name $I_t$. 
\end{definition}

\begin{figure}
\centering
\includegraphics{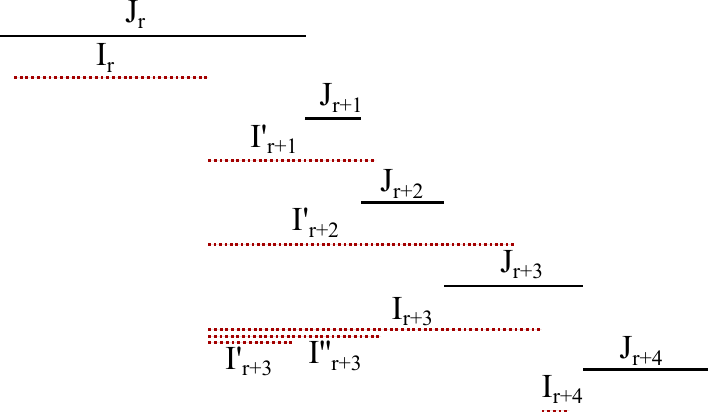}
\caption{An example of the matching process. Interval $J_r$ is nice. Interval $I'_{r+1}$
  is created from $I_r$ by Rule 4. It extends beyond $J_{r+1}$, so we leave $J_{r+1}$ unmatched. Interval $I'_{r+2}$ is created
  from $I'_{r+1}$ by Rule 3. Similarly, we leave $J_{r+2}$ unmatched. Interval $I'_{r+3}$ is created from
  $I'_{r+2}$ by Rule 3; it produces the sequence of candidate intervals $I'_{r+3},I''_{r+3},\cdots,I_{r+3}$ by repeated applications of
  Rule 1. Interval $J_{r+3}$ is matched to $I_{r+3}$. Intervals $J_{r+4}$ and $I_{r+4}$ exhibit a case in which the matched intervals do not 
  have to intersect. }
\label{fig:matching}
\end{figure}

\subsection{Main argument}
\label{sec:main-argument}

The set of nice intervals play central role in relating the revenue of $A$ and $B$. Before we proceed with the central argument, we give some essential definitions and prove some useful facts about the matched intervals.

\begin{definition}
Let $R(p,v)$ denote the total revenue extracted from the sequence of
values $v$ with prices $p$. Let $R_k(J,v)$ denote the
revenue that can be extracted from an interval $J$ of $v$ to which we
offer the $k$-level price $p_{(k)}$: $R_k(J,v)=p_{(k)} \cdot \left|\{i\,:\, i\in J \wedge 
v_i\geq p_{(k)}\}\right|$. In particular, $R_k(J,\vals)$ is equal to the mass of $J$ multiplied by $p_{(k)}$.
\end{definition}

\begin{lemma} \label{lem:matchingproperties}
The matching algorithm has the following properties:
\begin{enumerate}
\item If $J_t$ is nice, then $\mass_t(J_t)\leq  (1+2\epsilon) \mass_t(I_t)$.
\item For every constant $c_3>2$ let 
$$c_4=\frac{c_3}{c_3-2} \cdot(2+\epsilon) \cdot\epsilon.$$ 
If $J_t$ is not nice, then at least one of the following holds
\begin{align}
\label{eq:3}
\mass_t(J_t) &\leq c_3 \cdot \mass_t(I'_t) \\
\mass_t(J_t) &\leq c_4 \cdot \mass_t(I_t)
\end{align}
\end{enumerate}
\end{lemma}
\begin{proof}
The properties follow directly from the definitions and the process. For example, the first property, follows by the maximality of $I_t$ taking into account Rules 1 and 5; this is the only place in the proof where we use Rule 5. 

The second property is more involved and uses the fact that $J_t$ is not nice. Notice first that if $J_t$ is not matched, then $\mass_t(J_t)\leq \mass_t(I'_t) \leq c_3 \mass_t(I'_t)$ and the property holds. So now assume that $J_t$ is matched and
consider the sequence of intervals $K_0, K_1,\ldots, K_{r+1}$, produced when we start with $K_0=I'_t$ and repeatedly apply Rule 2; the sequence ends when the produced interval extends beyond $J_t$; that is, $\start(K_{r+1})<\e(J_t)<\e(K_{r+1})$ (See Figure~\ref{fig:structuralproperties}).

\begin{figure}
\centering
\includegraphics{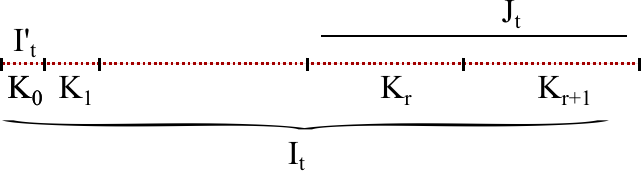}
\caption{Illustrating the proof of Lemma~\ref{lem:matchingproperties}.}
\label{fig:structuralproperties}
\end{figure}

Since $J_t$ is not nice, it does not include any of these intervals and in particular interval $K_r$, so we must have that $\mass_t(J_t)\leq \mass_t(K_r)+\mass_t(K_{r+1})$. This bounds the mass of $J_t$ by
$$\mass_t(J_t)\leq \left((1+\epsilon)^r+(1+\epsilon)^{r+1}\right)\mass_t(I'_t) =(2+\epsilon)(1+\epsilon)^r \mass_t(I'_t).$$

On the other hand, $I_t$ must contain almost all the mass of intervals $K_0,\ldots,K_r$. In particular, by the maximality of $I_t$, if we apply Rule 1 to $I_t$, we get an interval which extends beyond $J_t$, and it certainly contains the intervals $K_0,\ldots,K_r$. Thus the mass of $I_t$ is at least
\begin{align*}
\mass_t(I_t) &\geq \left(\mass_t(K_0)+\cdots+\mass_t(K_r)\right)/(1+\epsilon) \\ &= \left(1+(1+\epsilon)+\cdots+(1+\epsilon)^r\right)/(1+\epsilon) \cdot \mass_t(I'_t) \\
&=\left((1+\epsilon)^{r+1}-1\right)/\left(\epsilon(1+\epsilon)\right) \cdot\mass_t(I'_t).
\end{align*}

The lemma now follows from the above two bounds. Specifically, if $\mass_t(J_t)>c_3 \mass_t(I'_t)$, then $c_3\geq (2+\epsilon)(1+\epsilon)^r$. We can then bound 
\begin{align*}
\mass_t(J_t) &\leq (2+\epsilon)(1+\epsilon)^r \mass_t(I'_t) \\
& \leq (2+\epsilon)(1+\epsilon)^r \frac{\epsilon(1+\epsilon)}{(1+\epsilon)^{r+1}-1}\mass_t(I_t) \\
& \leq \frac{c_3}{c_3-2} \cdot(2+\epsilon) \cdot \epsilon \cdot \mass_t(I_t).
\end{align*}

\end{proof}

The essential part of proving the constant approximation ratio is a
charging argument. The idea is to charge each of the values $R_t(J_t, \vals)$ to
some value $R_s(I_s, \vals)$ of some nice interval $J_s$. To do this, we process the
intervals from left to right in phases. A phase begins
with a nice interval and ends exactly before the next nice interval (or at the
last interval). 

\begin{lemma}[Phase]
Let $s,\ldots s'-1$ be a phase, that is, $J_s$ is a nice interval and
none of the intervals $J_r$ for $r=s+1,\ldots, s'-1$ is nice. Then there is a contant $c$ such that
\begin{align}
\sum_{t=s}^{s'-1} R_t(J_t, \vals) &\leq c \cdot R_s(I_s, \vals) 
\end{align}
\end{lemma}

\begin{proof}
We charge each interval $J_t$, $t=s+1,\ldots,s'-1$ to some primal interval, according to the cases of the previous lemma, as follows:
\begin{enumerate}
\item If $\mass_t(J_t) \leq c_3 \cdot \mass_t(I'_t)$, we charge it to the previous matched interval. Specifically, if  $\pre{t}$ denotes the previous matched interval (for example, in Figure~\ref{fig:matching}, $\pre{{r+2}}=r$), then
$$R_t(J_t,\vals) \leq c_3 R_t(I'_t, \vals) \leq c_3 \left(\frac{1+\epsilon}{w}\right)^{t-\pre{t}} R_{\pre{t}}(I_{\pre{t}}, \vals).$$
\item Otherwise we charge $J_t$ to $I_t$. Specifically, by Lemma~\ref{lem:matchingproperties} we have 
$$R_t(J_t,\vals)\leq c_4 R_t(I_t,\vals).$$
\end{enumerate}

If we sum the above for $t=s+1,\ldots,s'-1$, we get
\begin{equation}
\label{eq:main:1}
\sum_{t=s+1}^{s'-1} R_t(J_t,\vals) \leq \sum_{t=s}^{s'-1} \alpha_t \cdot R_t(I_t,\vals),
\end{equation}
for some $\alpha_t$ which is the sum of all the above charges. The
crucial property of the parameters is that the total charge to every interval is low (less than $1/3$) in
every posible scenario. Specifically, the total charge to $\alpha_t$
can be at most $c_4$  (Case 2), plus at most $c_3\cdot\sum_{i=1}^{\infty}
\left(\frac{1+\epsilon}{w}\right)^i=c_3\cdot \frac{1+\epsilon}{w-1-\epsilon}$ from Case 1. The total charge is
\begin{equation}
\label{eq:main:alpha}
\alpha_t\leq \alpha= c_4+c_3\cdot\frac{1+\epsilon}{w-1-\epsilon},
\end{equation}
which is less than $1/3$ if we select the parameters appropriately.

We now use the optimality of price vector $p$ to bound $\sum_{t=s}^{s'-1} R_t(J_t,\vals)$.  Consider decreasing the prices in the interval $\cup_{t=s}^{s'-1} J_t$  to level $s'-1$. By
  the optimality of the price profile $p$, this cannot give better revenue for $\vals^A$, and
  therefore
\begin{equation} \label{eq:main:2x}
\sum_{t=s}^{s'-1} R_t(I_t,\vals^A)\leq \sum_{t=s}^{s'-1}
R_t(J_t,\vals^A).
\end{equation}
We want now to combine \eqref{eq:main:1} and \eqref{eq:main:2x}, but one involves $\vals$ and the other $\vals^A$, so we employ the balance property to transform \eqref{eq:main:2x} to:
\begin{equation}
\label{eq:main:2}
\sum_{t=s}^{s'-1} R_t(I_t,\vals) \leq 3 \sum_{t=s}^{s'-1} R_t(I_t,\vals^A)\leq 3 \sum_{t=s}^{s'-1} R_t(J_t,\vals^A) \leq 3 \sum_{t=s}^{s'-1} R_t(J_t,\vals).
\end{equation}

We will also need the first property of Lemma~\ref{lem:matchingproperties}: 
\begin{equation}
\label{eq:main:3}
R_s(J_s,\vals)\leq (1+2\epsilon) R_s(I_s, \vals).
\end{equation}
If we now combine Equations \eqref{eq:main:1}, \eqref{eq:main:2}, and \eqref{eq:main:3} (multiply \eqref{eq:main:2} by $\alpha$, take their sum, and divide the result by $1-3\alpha$), we get
\begin{align*}
\sum_{t=s}^{s'-1} R_t(J_t,\vals) \leq \frac{1+2\epsilon}{1-3\alpha} R_s(J_s,\vals)
\end{align*}
(this is were it is crucial that $\alpha$ is less than $1/3$). Letting now $c=\frac{1+2\epsilon}{1-3\alpha}$, we get the lemma.
\end{proof}

By simply summing the inequalities of the previous lemma for all nice intervals, we get:

\begin{lemma} \label{lemma:main}
Let $J_{q_0},\ldots, J_{q_{m'}}$ be the nice intervals. Assuming that the first interval is matched (Definition \ref{def:assumption2}), or equivalently that $q_0=0$, there is a constant $c$ such that
\begin{align} 
\sum_{j=0}^m R_j(J_j,\vals) \leq c\cdot\sum_{k=0}^{m'} R_{q_k}(I_{q_k},\vals).
\end{align}
\end{lemma}
\begin{figure}
\centering
\includegraphics{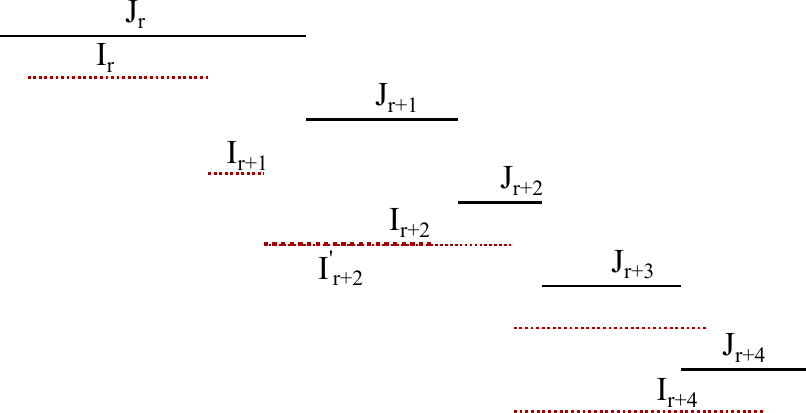}
\caption{An example of possible charging. $J_r$ is a nice interval. $J_{r+1}$ is charged either to $I_r$ (if the mass of $I'_{r+1}$ is small) or to $I_{r+1}$ (otherwise). Similarly $J_{r+2}$ is charged either to $I_{r+1}$ or to $I_{r+2}$ depending on the mass of $I'_{r+2}$ (shown). $J_{r+3}$ is charged to $I_{r+2}$ (because it is not matched). $J_{r+4}$ is
  charged either to $I_{r+2}$ or to $I_{r+4}$, depending on the mass of $I'_{r+4}$.}
\label{fig:main}
\end{figure}

The previous lemma relates the revenue from the whole sequence to the
revenue from the matching primal intervals of nice intervals.  Because on every primal balanced interval the revenues of $\vals$, $\vals^A$, and $\vals^B$ are all within a constant factor, we can immediately relate $R(p,\vals^A)$ to $R(p,\vals^B)$.
\begin{lemma} \label{lemma:AvsB}
With the assumptions of the previous lemma,
$$R(p,\vals^A) \leq 3c\cdot R(p,\vals^B).$$
\end{lemma}
\begin{proof}
We have the following derivations 
\begin{align*}
R(p,\vals^A) &= \sum_{j=0}^m R_j(J_j,\vals^A) \\
& \leq \sum_{j=0}^m R_j(J_j,\vals) \\
& \leq c \cdot \sum_{k=0}^{m'} R_{q_k}(I_{q_k},\vals) \\
& \leq 3c \cdot \sum_{k=0}^{m'} R_{q_k}(I_{q_k},\vals^B) \\
& \leq 3c \cdot \sum_{k=0}^{m'} R_{q_k}(J_{q_k},\vals^B) \\
& \leq 3c \cdot \sum_{j=0}^{m} R_j(J_j,\vals^B) \\
& \leq 3c \cdot R(p,\vals^B)
\end{align*}
The first inequality is trivial because $\vals^A$  has smaller values than $\vals$, the second inequality comes from the main lemma (Lemma~\ref{lemma:main}), the third inequality comes from the balance condition, the fourth inequality is based on the fact that $I_{q_k}$ is contained in $J_{q_k}$.
\end{proof}

\subsection{Stitching everything together}
\label{sec:stich-everyth-togeth}

In this subsection, we revisit and remove the First Running Assumption. For simplicity, we keep the Second Running Assumption and we deal with it later.

\begin{lemma}\label{l:ea}
For every valuation profile~$\vals$, 
$$\Pr\left[ R(p,\vals^A) \ge \frac{1}{3w} \cdot \M(\vals) \right] \ge \tfrac{1}{16}.$$
\end{lemma}
\begin{proof} 
Let $\pricesR^*$ be the optimal prices for $\M(\vals)$.
With probability~$1/4$, the bidders with the highest and
second-highest valuations lie in~$A$, i.e., $\vals^{(2)} = {\vals^A}^{(2)}$.   
Given this event, the conditional
expected revenue from bidders in~$A$ and~$B$ 
under the price vector $\pricesR^*$ is at least $\M(\vals)/2 $ and at
most~$\M(\vals)/2$, respectively.  By Markov's inequality,
the conditional expected revenue from bidders in~$A$ under~$\pricesR^*$
is at least $\tfrac{1}{3} \M(\vals)$ with probability at
least~$\tfrac{1}{4}$.

The optimal revenue $R(p,\vals^A)$ from $A$ cannot be worse than the revenue from the pricing vector $\pricesR^*$ scaled down to the next power $\frac{\vals^{(2)}}{w^k}$. Since with this scaling down, we loose at most a factor of $w$ in the revenue, we conclude that with probability at least $\tfrac{1}{16}$, $R(p, \vals^A) \ge \M(\vals)/3w$.
\end{proof}

If we choose appropriately the constant $s_0$, 
Lemma \ref{lemma:balance} ensures that for every constant
$\epsilon\in (0,1]$, and every valuation profile~$\vals$, every primal interval is balanced with probability at least $31/32$. (Note that for small $n$, this is vacuously true since there are no primal intervals).

The probability that all positive events (that the two higher values are in $A$, that $\M(\vals^A) \ge \frac{1}{3w} \cdot \M(\vals)$, and that all primal intervals are balanced) hold is at least $1 - \left( \frac{1}{32} + \frac{15}{16}\right)= \frac{1}{32}$.  So, in the analysis of the performance of the \algfull auction, we can assume that these events hold, which will increase the approximation ratio by a factor of $32$. 

By Lemma~\ref{lemma:AvsB}, we have that there are parameters $\epsilon$, $w$ and $s_0$ such that
$$R(p,\vals^B) \geq \frac{1}{3c} R(p, \vals^A) \geq \frac{1}{3c}
\frac{1}{3w} \M(\vals)= c_1 \cdot \M(\vals) ,$$
(where $c_1=1/(9cw)$), which establishes the theorem under the assumption that the first interval $J_0$ is nice.

\subsection{Padding}
\label{sec:intervals-1}

In the previous subsection, we analyzed the algorithm and bounded its
revenue under the assumption that the first interval $J_0$ is nice, i.e., it contains a
primal interval. We now show how to remove this assumption by padding
$\vals$ on the left with a sequence of $\vals^{(2)}$ values which will guarantee
that the assumption holds. Then it suffices to show that the
padding does not affect the revenue significantly.

In particular, let $\hv$ be the sequence consisting of $s'_0=10 s_0$ values
$\vals^{(2)}$ followed by $\vals$. With probability very close to 1, at
least $s_0$ of these values go to $\hv^A$, and therefore the interval
$J_0$ contains the initial primal interval\footnote{We need to select $s_0'$ large enough so that this probability is greater than $1-1/32$ to guarantee that all assumptions and events of the previous subsection hold with posititive probability.}. Notice that
$\M(\hv)=\M(\vals)+s'_0 \cdot
\vals^{(2)}$ because every optimal decreasing pricing can sell to the
first $s_0$ bidders at price $\vals^{(2)}$ without affecting the pricing
for the remaining bidders; this is because the prices for the
remaining bidders cannot exceed $\vals^{(2)}$. Similar constraints hold for
$\M(\hv^A)$ and $\M(\hv^B)$. 

Let $\algo_M$ be the expected revenue when we apply the pricing scheme
to $\vals^A$ and $\vals^B$ and let $\algo_F$ be the expected revenue when we
apply an algorithm which is competitive against $\F2$. The
expected revenue of the algorithm is $(\algo_M+\algo_F)/2$.  The
revenue extracted from the part against $\F2$ will play a crucial
role in relating the revenue of the original sequence with the revenue
with the padded sequence. For the padded sequence, let $\halgo_M$ be the
corresponding expected revenue for the padded sequence, that is, when
we apply the pricing scheme to $\vals^\hA$ and $\vals^\hB$.  From the previous
subsection, we have that
$$\halgo_M\geq c_1\cdot\left( \M(\vals)+s'_0 \cdot \vals^{(2)}\right),$$
for some constant $c_1$. Also, it is easy to see that the prices of the algorithm applied
to $\hv$ match the prices of
the algorithm applied to $\vals$, so we have that $\halgo_M \leq \algo_M
+s'_0 \cdot
\vals^{(2)}$. Putting the last two inequalities together, we get
\begin{align*}
\algo_M \geq c_1\cdot\left( \M(\vals)+s'_0 \cdot \vals^{(2)}\right)
- s'_0
\cdot \vals^{(2)}  \\ = c_1\cdot \M(\vals)- (1-c_1)\cdot s'_0 \cdot \vals^{(2)}.
\end{align*}
The
second term is the reason for targeting both $\F2$ and $\M$
and explains the need to run a competitive algorithm against $\F2$
with probability $1/2$.

To account for the second term, we observe that with probability
$1/2$, the algorithm gets revenue at least $c_1'\cdot \F2(\vals)\geq c_1' \cdot 2
\vals^{(2)}$, for some constant $c_1'$. That is,
$$\algo_F\geq c_1' \cdot 2 \vals^{(2)}.$$
Combining the two bounds for $\algo$, we get
\begin{align*}
\algo& \geq \frac{1}{2} \algo_M \\
\algo& \geq \frac{1}{2} \algo_F \\
\algo& \geq \frac{(1-c_1)\cdot s'_0\cdot
  \left(\frac{1}{2} \algo_F\right)+2c_1' \cdot \left(\frac{1}{2} \algo_M
  \right)}{(1-c_1)\cdot s'_0+2c_1'} \\
&= \frac{c_1\cdot c_1'}{(1-c_1)\cdot s'_0+2c_1'} \cdot
\M(\vals),
\end{align*}
which shows that the algorithm has constant approximation ratio.

\section{Multi-Unit Auctions}

In this section we extend our results to multi-unit auctions with
limited supply.  To develop this theory, we extend the monotone price
benchmark $\M$ to the case of an arbitrary number~$k \ge 2$ of units for sale.
We call a price vector~$\prices$ {\em feasible} for the
valuation profile~$\vals$ and supply limit~$k$ if: (i) $p_1 \ge p_2
\ge \cdots \ge p_n$; 
(ii) all prices are at most the second-highest valuation of~$\vals$;
and (iii) there are at most~$k$ bidders~$i$ with $\vali > \pricei$.
We allow our benchmark to break ties in an optimal way.  More
precisely, the revenue earned by a feasible price vector is $\sum_{i
  \,:\, \vali > \pricei} \pricei$ plus, if there are~$\ell$ items
remaining, the sum of the prices offered to up to~$\ell$ bidders~$i$
with $\vali = \pricei$.  We define the {\em $k$-unit monotone price
  benchmark} $\MK(\vals)$ as the maximum revenue obtained by a
price vector that is feasible for~$\vals$ and~$k$.

There are two main issues to address.
The first issue is to identify a class of priors ${\cal F}_i$
such that  $\MK(\vals)$  is a meaningful  benchmark for prior-free
approximation, i.e., it simultaneously approximates all optimal
auctions in multi-unit Bayesian settings with priors drawn from the class.
The challenge, relative to the unlimited-supply setting introduced
  in~\cite{LR12}, is that limited-supply Bayesian optimal auctions
  exhibit more complex behavior than unlimited-supply ones.
  Section~\ref{ss:just} shows, essentially, that the
  benchmark~$\MK(\vals)$ is meaningful for any valuation
  distributions that have pointwise ordered ironed virtual valuations
The second issue is to design auctions competitive with the
benchmark~$\MK(\vals)$.
We accomplish this through a general reduction, showing how to build a
limited-supply auction that is~$O(1)$-competitive w.r.t.~$\MK(\vals)$
from a digital goods auction that is~$O(1)$-competitive w.r.t.~$\M$.

\subsection{Justifying the $k$-Unit Monotone Price Benchmark}
\label{ss:just}

The goal of this section is to prove that every prior-free auction
that is~$O(1)$-competitive with the benchmark~$\MK(\vals)$ has
expected revenue at least a constant fraction of optimal in every
Bayesian multi-unit environment with valuation distributions lying in
a prescribed class.  Making this precise requires some terminology and
facts from the theory of Bayesian optimal auction design, as developed
by Myerson~\cite{M81}.  See also the exposition by Hartline~\cite{omd}.

Consider a bidder with valuation drawn from a prior distribution~$\dist$
with positive and continuous density~$f$ on some interval.  The {\em
  virtual   value}~$v$ at a point~$v$ in the support is defined as
$$\phi(v)= v - \frac{1 - {\cal F}(v)}{f(v)}.$$
For example, if~$\dist$ is the uniform distribution on $[0,a]$, then the
corresponding virtual valuation function is $\phi(v) = 2v-a$.

For clarity, we first discuss the case of {\em regular} distributions,
meaning distributions with nondecreasing virtual valuation functions.
In this case, the Bayesian optimal auction awards items to the (at
most~$k$) bidders with the highest positive virtual valuations.  The
payment of a winning bidder is the minimum bid at which it would
continue to win (keeping others' bids the same).  That is, if the
$(k+1)$th highest virtual valuation is~$z$, then every winning
bidder~$i$ pays~$\phi_i^{-1}(\max\{0,z\})$.  For these prices to be
related to the monotone price benchmark, we need to impose
conditions on the $\phi_i^{-1}(z)$'s.  This contrasts with
unlimited-supply settings, where restricting the~$\phi^{-1}_i(0)$'s
--- that is, the monopoly reserve prices --- to be nonincreasing
in~$i$ is enough to justify the monotone-price benchmark~\cite{LR12}.  
Since the $(k+1)$th highest virtual valuation could be anything, the
natural extension of the condition in~\cite{LR12} is to
restrict~$\phi^{-1}_i(z)$ to be nonincreasing in~$i$ for every
non-negative number~$z$.

Accommodating irregular distributions, for which the optimal Bayesian
auction is more complicated, presents additional
complications.  Each virtual valuation function~$\phi_i$ is
replaced by the ``nearest nondecreasing approximation'', called the
{\em ironed virtual valuation function}~$\ivvi$.  The optimal auction
awards the items to the (at most~$k$) bidders with the highest
positive ironed virtual valuations.  Since ironed virtual valuation
functions typically have non-trivial constant regions, ties can occur,
and we assume that ties are broken randomly.  That is, if
there are~$k$ items, a group~$S$ of bidders with identical ironed
virtual values~$z > 0$, $\ell < k$ bidders with ironed virtual value
greater than~$z$, and~$\ell+|S| > k$, then $k-\ell$ winners from~$S$
are chosen uniformly at random.

We call valuation distributions~$\dist_1,\ldots,\dist_n$ {\em pointwise
  ordered} if $\ivvi^{-1}(z)$ is nonincreasing in~$i$ for every
  non-negative~$z$.\footnote{Since~$\ivvi$ is continuous and
  nondecreasing, $\ivvi^{-1}(z)$ is an interval.  If the inverse image
  has multiple points, we define $\ivvi^{-1}(z)$ by the infimum.  If
  the inverse image is empty, we define $\ivvi^{-1}(z)$ as the left or
  right endpoint of the distribution's support, as appropriate.}
The motivating parametric examples discussed
  earlier --- uniform distributions with intervals~$[0,h_i]$ and
nonincreasing $h_i$'s, exponential 
distributions with nondecreasing rates, and Gaussian distributions with
nonincreasing means --- are pointwise ordered in this sense.

We also require a second condition, which we inherit from the standard
i.i.d.\ unlimited-supply setting.  The issue is that, with arbitrary
irregular distributions, no prior-free auction can be simultaneously
near-optimal in all Bayesian environments, even with i.i.d.\ bidders
and unlimited supply.\footnote{Informally, consider valuation
distributions that 
take on only two values, one very large (say~$M$) and the other~0.
Suppose the probability of having a very large valuation is very small
(say~$1/n^2$).  If the distribution is known, the optimal auction uses
a reserve price of~$M$ for each bidder.  Elementary arguments, as
in~\cite{HR08}, show that no single auction is near-optimal for all
values of~$M$.}  Various mild conditions are sufficient to rule out
this problem; see~\cite{HR08} for a discussion.  Here, for simplicity,
we restrict attention to {\em well-behaved} Bayesian multi-unit
environments, meaning that the Bayesian optimal auction derives at
most a constant fraction (90\%, say) of its revenue from outcomes 
in which some winner is charged a price higher than the second-highest
valuation.  (Such a winner is necessarily the bidder with the highest
valuation.)  Standard distributions always yield
well-behaved environments.  Even pathological distributions produce
well-behaved environments provided the market is sufficiently large
(e.g., there are enough bidders drawn i.i.d.\ from each of the
distributions).

Our main result in this section
 is that approximating the $k$-unit monotone price
benchmark guarantees simultaneous approximation of the optimal auction
in all well-behaved Bayesian multi-unit environments with pointwise
ordered distributions.

\begin{theorem}
\label{thm:irregular-prior-free}
If the expected revenue of the multi-unit auction~$\A$ is at least a
constant fraction of~$\MK(\vals)$ on every input, then, in every
well-behaved multi-unit Bayesian environment with pointwise ordered
distributions, the expected revenue of~$\A$ is at least a constant
fraction of that of the optimal auction for the environment.
\end{theorem}

\begin{proof}
Fix an auction that is~$\beta$-competitive with~$\MK(\vals)$ on every
input.  Fix a well-behaved Bayesian multi-unit environment with
pointwise ordered valuation distributions~$\dist_1,\ldots,\dist_n$.
Let~$\A^*$ be the optimal auction for this environment.
We claim that, for every input~$\vals$ in which the revenue collected
by~$\A^*$ from the bidder with the highest valuation is at most the
second-highest valuation, the benchmark~$\MK(\vals)$ is at
least half the expected revenue of~$\A^*$ on~$\vals$.  This implies
that the expected revenue of~$\A$ is at least~$1/2\beta$ times that
of~$\A^*$ on this input.  Since the environment is well behaved, the
theorem follows.

To prove the claim,
fix an input~$\vals$, as above.  Recall that~$\A^*$, as a Bayesian
optimal auction, awards items to the (at most~$k$) bidders with the
highest positive ironed virtual valuations, breaking ties randomly.
The tricky case of the proof is when ties occur.  Assume there are~$k$
items, a group~$S$ of bidders with common ironed virtual value~$z >
0$, and a group~$T$ of $\ell \in (k-|S|,k)$ bidders with ironed
virtual value greater 
than~$\ell$ (so $|S| > k-\ell$).  We next explicitly compute the
payments collected by~$\A^*$ on this input, using the standard payment
formula for incentive-compatible mechanisms
(see~\cite{M81} or~\cite{omd}).  Let~$a_i$ and~$b_i$ denote
the left and right endpoints, respectively, of the interval of
values~$v$ that satisfy~$\ivvi(v) = z$.  Since the distributions are
pointwise ordered, the $a_i$'s and the $b_i$'s are nonincreasing
in~$i$.  Let~$q = (k-\ell)/|S|$ denote the winning probability of a
bidder in~$S$.  Define~$q' = (k-\ell+1)/(|S|+1)$ as the hypothetical
winning probability of a bidder in~$T$ if it lowered its bid to the
value~$\ivvi^{-1}(z)$.
The expected payment of a bidder~$i$ in~$S$ is~$qa_i$
(i.e., $a_i$ in the event that it wins).  The expected payment of a
bidder~$i$ in~$T$ (who wins with certainty) is~$q'a_i + (1-q')b_i$.
To complete the proof, we argue that $\MK(\vals)$ is at least the
expected revenue collected by~$\A^*$ from the bidders in~$S$, and
also at least that from the bidders in~$T$.

Projecting onto a subset of bidders only decreases the value of the
$k$-unit monotone price benchmark~$\MK(\vals)$ (see
Lemma~\ref{app:project} for the formal argument).
First, project onto
the~$k$ bidders of~$S$ with the highest~$a_i$ values.
Consider charging each such bidder the price~$a_i$.  This is a
monotone price vector.  By our assumption on the input~$\vals$, all of
these prices are at most the second-highest valuation in~$\vals$.
By the definitions, $\vali \ge a_i$ for every bidder $i \in S$ so
every offer will be accepted.  The resulting revenue is at least the
expected revenue earned by~$\A^*$ on~$\vals$,
and the value of the monotone price benchmark can only
be higher.  This shows that $\MK(\vals)$ is at least the expected
revenue collected by~$\A^*$ from bidders in~$S$.

Similarly, project onto the (at most~$k$) bidders of~$T$, and consider
charging each such bidder~$i$ the price $q'a_i + (1-q')b_i$.  Again,
this is a monotone price vector with all prices bounded above by the
second-highest valuation of~$\vals$, and every offer will be
accepted.  The value of the monotone price benchmark can only be
larger, so~$\MK(\vals)$ is also at least the expected revenue collected
by~$\A^*$ from bidders in~$T$.  The proof is complete.
\end{proof}

\subsection{Reduction from Limited to Unlimited Supply}\label{ss:reduction}


Having justified the $k$-unit monotone price benchmark~$\MK(\vals)$,
we turn to designing auctions that approximate it well.  We show that
competing with this benchmark reduces to competing with the
benchmark~$\M$ in unlimited-supply settings.  This reduction from
limited to unlimited supply is a generalization of one for identical
bidders~\cite{G+06}.  The idea is to first identify the~$k$ ``most
valuable'' bidders, and then run an unlimited-supply auction on them.
In contrast to the identical-bidder setting in~\cite{G+06}, the most
valuable bidders with an ordering are not necessarily those with the
highest valuations.  For example, a high-valuation bidder late in the
ordering need not be valuable, because extracting high revenue from it
might necessitate excluding many moderate-valuation bidders earlier in
the ordering.

We analyze the ``black-box reduction'' shown in Figure~\ref{fig:algopt}.

\begin{figure}[h]
\hrule\medskip
\textbf{Input}: A valuation profile $\vals$ for a totally ordered
set~$N = \{1,2,\ldots,n\}$ of bidders and $k$ identical items.  A
truthful digital goods (unlimited supply) auction~$\A$ for ordered
bidders. 
\begin{enumerate}

\item Let $\prices^*$ achieve the optimum monotone price
  benchmark~$\MK(\vals)$ for~$\vals$ and $k$.  Let $S = \{i\in
  N\,:\, v_i\geq p^*_i\}$ be the set of winners
  under~$\prices^*$.

\item Run the unlimited supply auction~$\A$ on the bidders~$S$, with
  the induced bidder ordering.

\item Charge suitable prices so that truthful reporting is a dominant
  strategy for every bidder.

\end{enumerate}
\caption{\textsf{The auction \algoptfull.}\label{fig:algopt}}
\medskip\hrule
\end{figure}

\begin{theorem}\label{thm:reduction}
If~$\A$ is a truthful unlimited-supply auction with ordered bidders
that is~$\beta$-competitive with~$\M$, then the \algoptfull auction is 
a truthful limited-supply auction with ordered bidders
that is~$2\beta$-competitive with~$\MK(\vals)$.
\end{theorem}

\begin{proof}
The \algoptfull auction is clearly feasible, in that there are always
at most~$k$ winners.  The first step can be implemented efficiently
using dynamic programming, so if~$\A$ runs in
polynomial time, then so does the \algoptfull auction.
To see that it is truthful, first note that the
allocation rule in the first step is monotone --- if bidder~$i$ belongs
to the computed set~$S$ in the profile~$\vals$, then it also belongs
to~$S$ for every bigger valuation~$\vali$ (holding other bidders'
valuations fixed).  The composition of this rule with the truthful
auction in the second step is also monotone (i.e., bidding higher only
increases winning probability), so there are unique payments that
render the auction truthful (see e.g.~\cite{omd}).  These payments are
easy to describe.  Monotonicity of the first step implies that each
bidder~$i$ faces a threshold bid $t_i(\vals_{-i})$ that is necessary
and sufficient to be included in the computed set~$S$.  The payment of
a winner~$i$ in the final step of the \algoptfull auction is simply
$t_i(\vals_{-i})$ or the payment computed by~$\A$, whichever is larger.

We prove the performance guarantee by arguing the following two
statements: (i) the unlimited supply benchmark~$\M$ applied to~$S$ is
at least half of the limited-supply benchmark~$\MK(\vals)$ applied to the
original bidder set; and (ii) the revenue of \algoptfull on the
original bidder set is at least that of the unlimited-supply
auction~$\A$ with the bidders~$S$.  The second statement follows
immediately from the facts that the winners of \algoptfull are the same
as those of~$\A$, and that the winners' payments are only higher.  For
statement~(i), consider prices~$\prices^*$ that
determine the benchmark~$\MK(\vals)$.  The projection~$\prices^*_S$
of this price 
vector onto the set~$S$ of bidders has revenue exactly~$\MK(\vals)$.
If $\prices^*_S$ is feasible, then it certifies that the
benchmark~$\M$ is at least~$\MK(\vals)$.  The only issue is
if the second-highest bidder is excluded from~$S$, in which
case $\prices^*_S$ might use a price larger than the second-highest
valuation in~$S$ (which is not permitted by the benchmark~$\M$).
But such a price can only extract revenue from the bidder with the
highest valuation, and every price of~$\prices^*$ is at most the
second-highest valuation~$\val^{(2)}$ of the original bidders.  Thus,
we can restore feasibility to~$\prices^*_S$ by lowering at most one
price to the second-highest valuation of~$S$, and we
lose revenue at most~$\val^{(2)}$.  Since~$\MK(\vals) \ge
2\val^{(2)}$ --- consider the price vector that offers~$\val^{(2)}$ to
everybody --- we retain at least half the revenue of~$\prices^*_S$.
Statement~(i) and the theorem follow.
\end{proof}

Of course, we can use the \algfull auction from
Section~\ref{sec:prior-free-auction} in Theorem~\ref{thm:reduction} to
obtain a truthful limited-supply auction that is~$O(1)$-competitive
with the benchmark~$\MK(\vals)$.
Theorem~\ref{thm:irregular-prior-free} implies that the resulting
auction also enjoys a strong simultaneous approximation guarantee in
Bayesian environments.

\bibliographystyle{plain}
\bibliography{m2}

\appendix

\section{Missing Proofs}

\begin{lemma}\label{app:project}
For every valuation profile~$\vals$, $k \ge 2$, and subset~$S$ of the
bidders with induced profile $\vals^S$, $\MK(\vals) \ge \M(k,\vals^S)$.
\end{lemma}

\begin{proof}
(Sketch.) Fix an input~$\vals$, with monotone prices~$\prices^*$
determining~$\MK(\vals)$.  
By induction, we only need to show that adding a single new bidder~$i$ can
only increase the value of the benchmark.  
Start by offering~$i$ the same price~$q$ as its predecessor in the
ordering (or the second-highest valuation, if there is no predecessor).
If~$i$ rejects (i.e., $\vali < q$), this extended price vector is
feasible and we are done (the optimal feasible price vector is only
better).  If~$i$ accepts (i.e., $\vali \ge q$) then
the price vector is infeasible (with $k+1$ winners) and we argue as
follows.  Go through the bidders after~$i$ one by one, increasing the
offer price to~$q$.  This preserves monotonicity.  If a previously
winning bidder ever
rejects this higher offer price, we are done (feasibility is restored
and the overall revenue is higher).  If not, there is now a ``suffix''
of bidders with the common offer price~$q$.  (This case only occurs
if~$i$ is after all of the winners in~$\prices^*$.)  We now increase
their common offer price until it equals that of the previous bidder,
thereby increasing the number of bidders in the suffix.  Eventually a
bidder that was winning under $\prices^*$ will reject the new offer
price (otherwise it would contradict the optimality of~$\prices^*$),
leaving us with a feasible monotone price vector with revenue at least
that of the original one.
\end{proof}

\end{document}